\documentclass{article}

\newcommand{\full}[1]{#1} 

\newcommand{\conf}[1]{}   

\usepackage{graphicx,amssymb,amsmath}
\usepackage{url,float}
\usepackage{amsfonts}
\usepackage{latexsym}
\usepackage[boxruled]{algorithm2e}
\usepackage{tcolorbox} 
\usepackage{enumerate} 

\usepackage{color}

\newcommand{\remove}[1]{{}}

\usepackage[normalem]{ulem}

\newcommand{\ABox}{
\raisebox{3pt}{\framebox[6pt]{\rule{6pt}{0pt}}}
}
\full{\newenvironment{proof}{{\bf Proof:}}{\hfill\ABox}}

\full{
\newtheorem{theorem}{{\bf Theorem}}

\newtheorem{lemma}{Lemma}
}

\newcommand{\lemlab}[1]{\label{lemma:#1}}
\newcommand{\thmlab}[1]{\label{thm:#1}}
\newcommand{\figlab}[1]{\label{fig:#1}}
\newcommand{\seclab}[1]{\label{sec:#1}}

\newcommand{\lemref}[1]{\ref{lemma:#1}}

\newcommand{\secref}[1]{\ref{sec:#1}}
\newcommand{\figref}[1]{\ref{fig:#1}}

\def\d{{\delta}}
\def\e{{\varepsilon}}

\def\q{{\theta}}
\def\f{{\phi}}

\def\bH{{\partial H}}

\def\R{{\mathbb{R}}}

\newcommand\butf{\operatorname{bf}}

\newcommand{\squeezelist}{\setlength{\itemsep}{0pt}}

\title{Threadable Curves}

\author{
Joseph O'Rourke and Emmely Rogers%
    \thanks{Department of Computer Science, Smith College, Northampton, MA, USA.
      \protect\url{{jorourke,erogers}@smith.edu}.}
      }

\index{Rogers, Emmely}
\index{O'Rourke, Joseph}

\begin{document}
\maketitle

\begin{abstract}
We define a plane curve to be \emph{threadable} if it can rigidly pass through a point-hole in 
a line $L$ without otherwise touching $L$.
Threadable curves are in a sense generalizations of monotone curves.
We have two main results.
The first is a linear-time algorithm for deciding whether a polygonal
curve is threadable---$O(n)$ for a curve of $n$ vertices---and if threadable, 
finding a sequence of rigid motions to thread it through a hole.
We also sketch an argument that shows that the threadability of algebraic curves
can be decided in time polynomial in the degree of the curve.
The second main result is an $O(n \operatorname{polylog} n)$-time
algorithm for deciding whether a 3D polygonal curve can thread through 
hole in a plane in $\mathbb{R}^3$,
and if so, providing a description of the rigid motions that achieve the threading.
\end{abstract}

\section{Introduction}
\seclab{Introduction}
We define a simple (non-self-intersecting) open planar curve $C$ to be \emph{threadable} 
if there exists a continuous sequence of rigid motions that allows
$C$ to pass through a point-hole $o$ in an infinite line $L$ without any other point of $C$
ever touching $L$. For fixed $L$, we will take $L$ to be the $x$-axis and $o$ to be the origin; equivalently we can view $C$ as fixed and $L$ moving (Lemma~\lemref{thread2line}).
$C$ could be a polygonal chain or a smooth curve.
$C$ is open in the sense that it is not closed to a cycle.
An example is shown in Fig.~\figref{line_35_62};
animations are available at
\url{http://cs.smith.edu/~jorourke/Threadable/}.
\begin{figure}[htbp]
\centering
\includegraphics[width=0.90\linewidth]{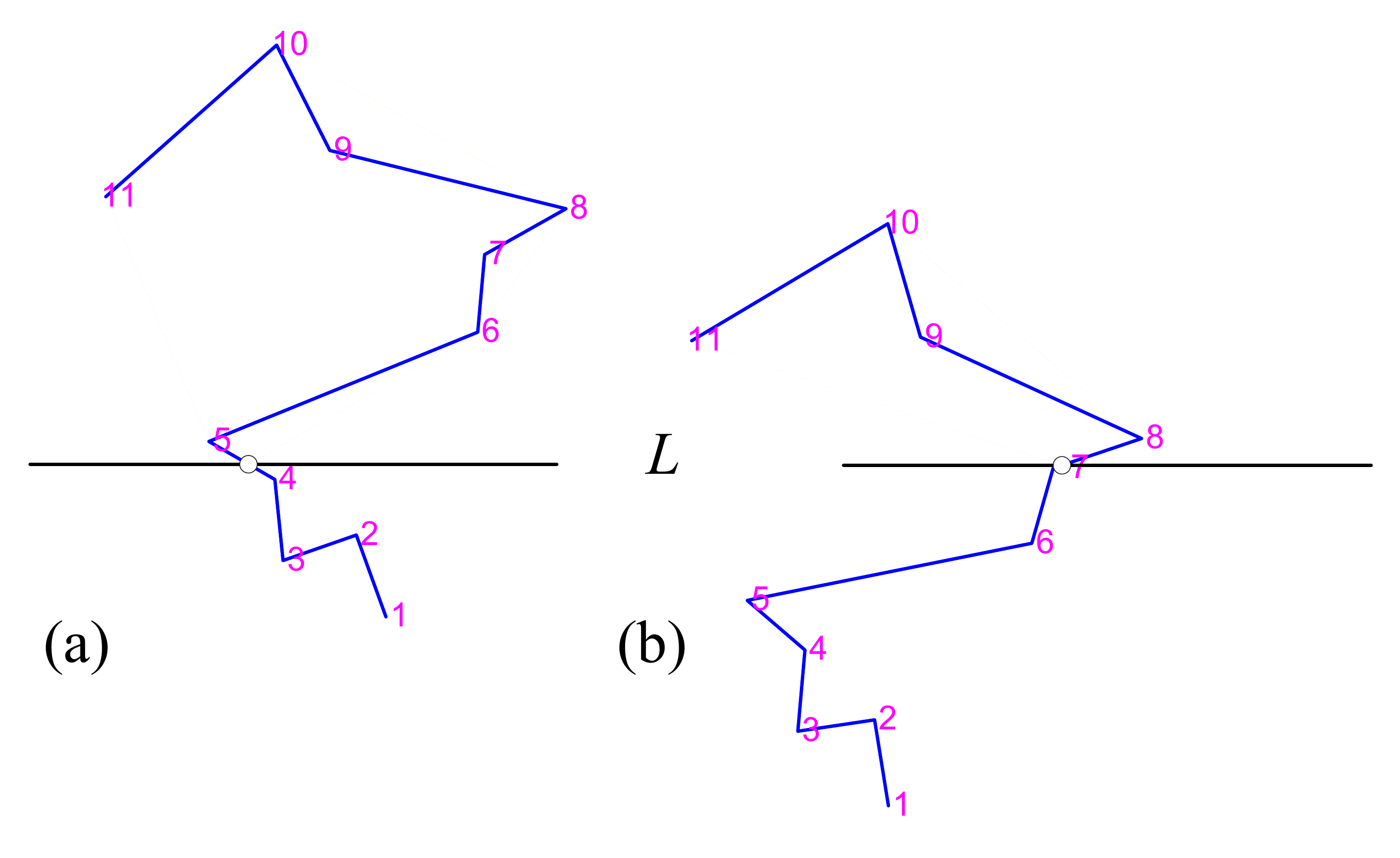}
\caption{Two snapshots of a $10$-segment polygonal chain passing through a point-hole
in the $x$-axis.}
\figlab{line_35_62}
\end{figure}

Note that our definition requires ``strict threadability'' in the sense that
no other point of $C$ touches $L$. So,  for example,
the curve illustrated in  
Fig.~\figref{ThreadableStrictDef} is not threadable.
\begin{figure}[htbp]
\centering
\full{\includegraphics[width=0.6\linewidth]{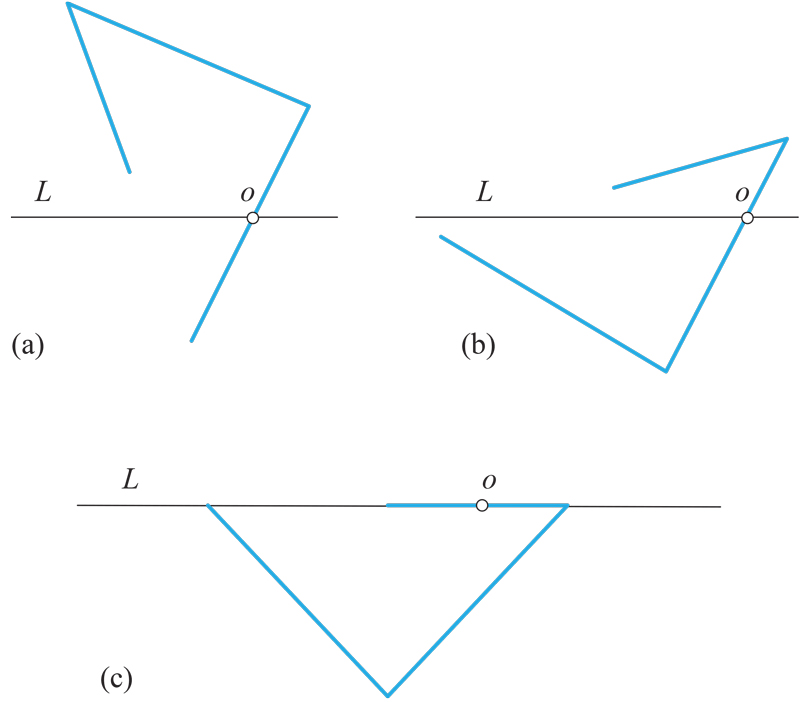}}
\conf{\includegraphics[width=0.8\linewidth]{Figures/ThreadableStrictDef}}
\caption{(a,b)~A curve that is not threadable: two snapshots partially through $o$.
(c)~To pass completely through $o$, 
an edge would have to lie on $L$.}
\figlab{ThreadableStrictDef}
\end{figure}

This notion has appeared in the literature in another guise.\footnote{
We thank Anna Lubiw for this reference.}
In particular, a threadable curve $C$ corresponds to a
``generalized self-approaching curve'' with width $\pi$ in both directions,
as defined in~\cite{aaiklr-gsac-01}.
However, those authors do not explore that concept,\footnote{
``One might also consider a symmetric situation, where curves are $\f$-self-approaching in both directions. Generalizations to 3D are also completely open.''}
and in any case, our explorations focus on different properties of $C$.

One could view our topic as a specialized motion-planning problem,
but it seems not directly addressed in the literature.
Work of Yap~\cite{yap1987move}, discussed in Section~\secref{HigherDim},
can be viewed as a higher dimensional version.
The research of
Arkin et al.~\cite{afm-asmpo-02}, examining the fabrication of
hydraulic tubes, leads to a concern for workspace clearance, to
which we return in Section~\secref{Open}.
We will see that classical computational geometry tools
suffice to address our problems, but some interesting questions are raised.

\subsection{Definition Consequences}
\seclab{DefinitionConsequences}
We now explore a few consequences of the definition.

\begin{lemma}
If a curve $C$ is threadable, then through every point $p \in C$
there is a line $L$ that meets $C$ in exactly $p$: $L \cap C = \{ p\}$,
and $L$ properly crosses $C$ at $p$.
\lemlab{thread2line}
\end{lemma}
\begin{proof}
This is a nearly immediate consequence of the definition, because at any one
time the $x$-axis serves as $L$, meeting $C$ at $p=o$.
So one can imagine $C$ fixed and $L$ undergoing rigid motions.
\end{proof}

Note that $L$ tangent to $C$ is insufficient for threadability, 
for then $C$ would locally lie on one side of $L$.
This is why the lemma insists on proper crossings.

What is perhaps not immediate is the implication in the other direction
to Lemma~\lemref{thread2line}:
\begin{lemma}
If a curve $C$ has the property that through every point $p \in C$
there is a line $L$ that meets $C$ in exactly $p$,
and $L$ properly crosses $C$ at $p$, then $C$ is threadable.
\lemlab{line2thread}
\end{lemma}

\noindent
The reason this is not immediate, is that 
it is conceivable that the orientation of the line changes discontinuously at
some point $p \in C$, requiring an instantaneous rigid ``jump'' motion of $C$ to
pass through $L$, rather than a continuous rigid motion.
A proof is deferred until we can rule out this discontinuity (Section~\secref{Hulls}).

\subsection{Monotone Curves}
\seclab{MonotoneCurves}
A \emph{monotone curve} $C$ is defined as one that meets all lines parallel to some
line $L$ in a single point (if \emph{strictly monotone}), or which intersects every line
parallel to $L$ in either a point or a segment (if \emph{non-strictly monotone}).
Every strictly monotone curve is threadable, and one can view threadability
as a generalization of monotonicity, allowing the orientation of $L$ to vary.

\section{Butterflies}
\seclab{Butterflies}
Define the \emph{butterfly} $\butf(p)$ for $p \in C$ to be the set of all lines $L$
satisfying the threadability condition at $p$: those lines
that meet $C$ in exactly $p$ and properly cross $C$ at $p$.
Let $L$ be one line in $\butf(p)$, and view $C$ as passing
through $L$ at $p$.
Then the convex hull $H^+$ of the chain from $p$ upward is
above $L$ and meets $L$ exactly at $p$, and the
hull $H^-$ of the chain from $p$ downward is
below $L$ and again meets $L$ exactly at $p$. 
(Here ``upward'' and ``downward" are not meant literally, but just convenient
shorthand for the two portions of the curve delimited by a roughly
horizontal $L$.)
If either hull met
$L$ in more than just $p$, then strict threadability would be violated at $L$.
Now rotate $L$ counterclockwise about $p$ until it hits $C$ at some point other
than $p$, and similarly clockwise.
The stopping points determine the butterfly \emph{wing-lines}.
See Fig.~\figref{bf_35_62}.
\begin{figure}[htbp]
\centering
\includegraphics[width=0.90\linewidth]{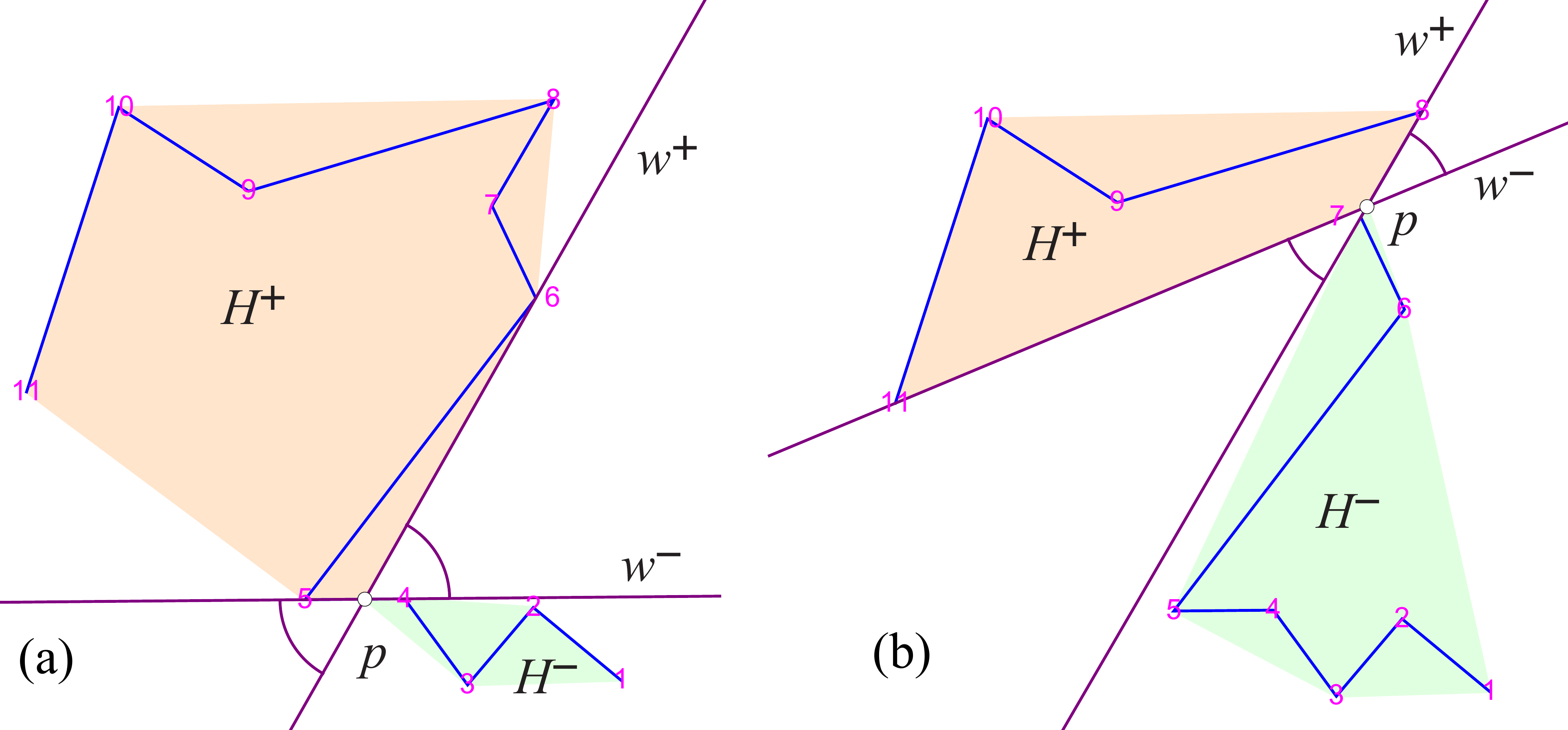}
\caption{Here $C$ is fixed, and two $\butf(p)$'s
are shown. Note the hulls $H^+$ and $H^-$ meet at exactly $p$.
(a)~The stopping point ccw is vertex $6$ and cw it is vertices $4,5$.}
\figlab{bf_35_62}
\end{figure}

Thus
$\butf(p)$ is an open double wedge. Its two boundary wing-lines
$w^+$ and $w^-$ 
(which are not part of $\butf(p)$) 
must both be externally supported by points of $C$ distinct from $p$.
Each wing must touch $C$ on at least one of its two halves with respect to $p$.
Note by our definition, $\butf(p)$ can never be a line; rather it becomes
empty when the wings-lines merge to one line.

\section{Upper and Lower Hulls}
\seclab{Hulls}
It is not difficult to see that the upper convex hull $H^+$ changes continuously
(say, under the Hausdorff distance measure) as $p$ moves along $C$, and similarly for $H^-$.
This has long been known in the work on computing ``kinetic'' convex hulls of continuously
moving points (although we have not found an explicit statement).
Roughly, because each point in the convex hull of a finite set of points is a convex combination 
of those points, moving one point $p$ a small amount $\e$ changes the hull
by at most a small amount $\d$.
For more detail, see~\cite{MSE1}.

Because the hulls change continuously,
the butterflies change continuously as well.
So we have finally established Lemma~\lemref{line2thread}:
If there is a line through every $p \in C$ meeting the threadability criteria,
then indeed $C$ is threadable: there are continuous rigid motions that
move $C$ through a point-hole in a line.

And now this is an immediate consequence of 
Lemma~\lemref{line2thread} and our definition of $\butf(p)$:
\begin{lemma}
A curve $C$ is threadable if and only if 
$\butf(p)$ is never empty for any $p \in C$.
\lemlab{bf-non-empty}
\end{lemma}

We can also now see this characterization, which is the basis of the algorithm
in the next section:
\begin{lemma}
A curve $C$ is threadable if and only if, for every $p \in C$,
the upper and lower hulls intersect in exactly $p$:
$H^+  \cap H^- = \{p\}$.
\lemlab{HullInt}
\end{lemma}
\begin{proof}
\begin{enumerate}[($\Rightarrow$)]
\item Suppose $C$ is threadable, but $H^+  \cap H^- \neq \{p\}$. 
We then show $C$ could not be threadable.
\begin{itemize}
\item{Case 1}: $H^+ \cap H^-$ is a 2D region (Fig.~\figref{Cases12}(a)).
Then $p$ is strictly interior to one of $H^+$ or $H^-$. So, the butterfly $= \varnothing$. 
Therefore $C$ is not threadable by Lemma~\lemref{bf-non-empty}.
\item{Case 2}: $H^+ \cap H^-$ is a segment (Fig.~\figref{Cases12}(b)). Note the intersection could not
consist of $\geq 2$ segments, for that would violate the convexity of convex hulls. 
So, the butterfly wings reduce to a line;
so the butterfly is empty. 
And again, $C$ is not threadable by Lemma~\lemref{bf-non-empty}.
\end{itemize}
\end{enumerate}

\begin{enumerate}[($\Leftarrow$)]
\item
Assume $H^+  \cap H^- = \{p\}$ for every $p$. Then, by the definition of $\butf(p)$, for every $p$ 
the butterfly is non-empty, because one could rotate a line through $p$ until it
hit $H^{\pm}$. So Lemma~\lemref{bf-non-empty} implies that $C$ is threadable.
\end{enumerate}
\end{proof}

\begin{figure}[htbp]
\centering
\includegraphics[width=0.75\linewidth]{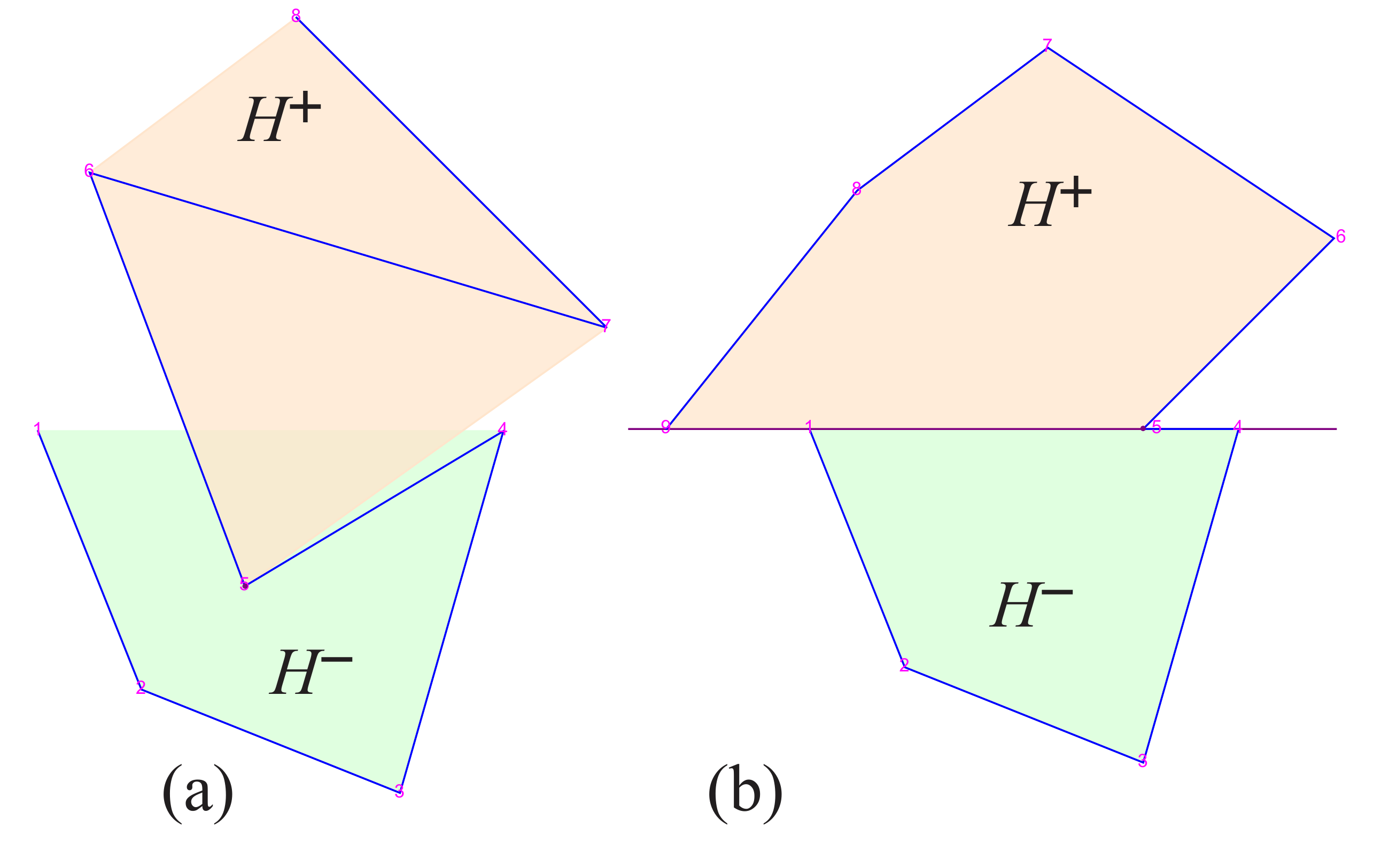}
\caption{(a)~An example of Case 1: $H^+ \cap H^-$ is a 2D region.
(b)~An example of Case 2: $H^+ \cap H^-$ is a segment.}
\figlab{Cases12}
\end{figure}


\section{Algorithm for Threadability}
\seclab{Algorithm}
In light of Lemma~\lemref{HullInt}, we can detect whether a polygonal chain
is threadable by computing $H^+$ and $H^-$ for all $p$ along $C$, and verifying
that $p$ never falls inside either hull.
Let $p$ be a point on $C=(v_1,v_2,\ldots,v_n)$, which we view as moving ``vertically
downward'' from $v_1$ (top) to $v_n$ (bottom).
Let the edges of $C$ be $e_i=(v_{i-1} v_i)$.
We concentrate on constructing $H=H^+$ as $p$ moves downward along $C$.
Clearly the same process can be repeated to construct $H^-$.

As $p$ moves down along $C$,
$H = \textrm{hull} \{ v_1, \ldots v_{i-1}, p \}$ grows in the sense
that the hulls form a nested sequence.
Thus once a vertex of $C$ leaves $\bH$, it never returns to $\bH$
(where $\bH$ is the boundary of $H$.)
At any one time, $p$ is a vertex of $H$.
Let $a_1, a_2$ be the vertices of $H$ right-adjacent to $p$, and
$b_1,b_2$ the vertices left-adjacent, so that $(b_2,b_1,p,a_1,a_2)$
are consecutive vertices of $H$.
Finally, let $A$ and $B$ be the lines through $a_1 a_2$ and $b_1 b_2$ respectively.
See Fig.~\figref{GrowingHull_4steps}.


\full{
\begin{figure}[htbp]
 \makebox[\textwidth][c]{%
\includegraphics[width=1.4\linewidth]{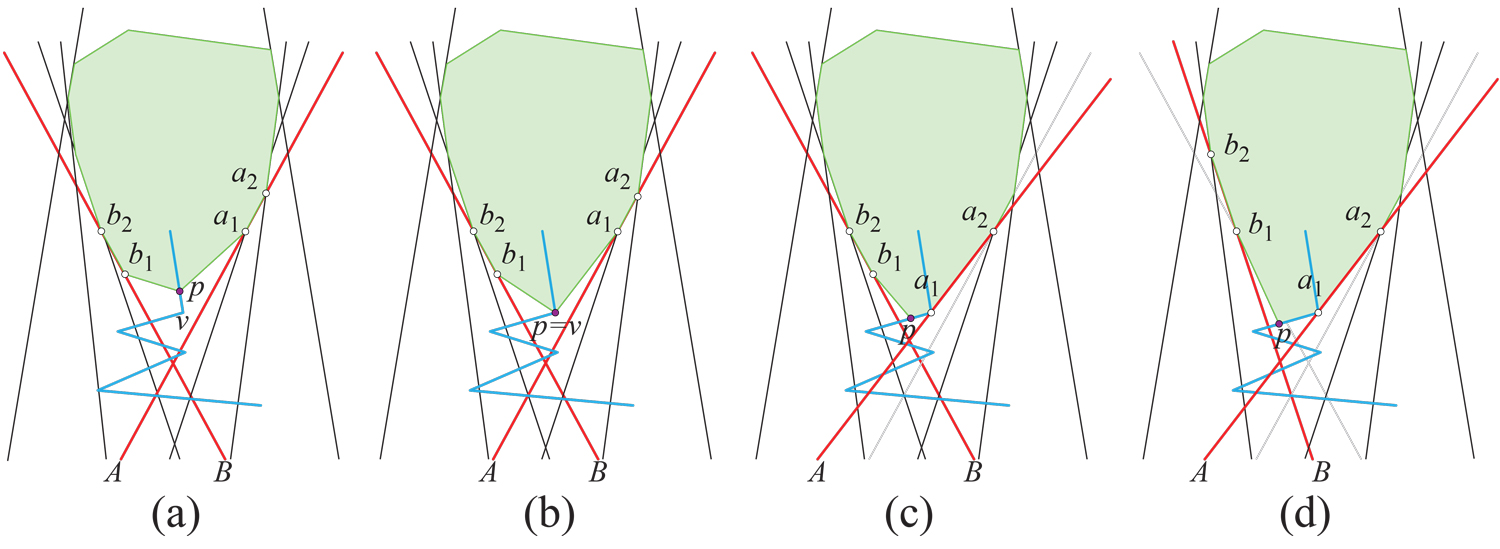}
}
\caption{Algorithm snapshots.
(a)~$H$ grows without combinatorial change until $p$ reaches $v$.
(b)~$p=v$ event.
(c)~$a_1,a_2$ updated. $e_i$ crosses $B$.
(d)~$b_1,b_2$ updated.}
\figlab{GrowingHull_4steps}
\end{figure}
}

\conf{
\begin{figure}[htbp]
\centering
\includegraphics[width=1.05\linewidth]{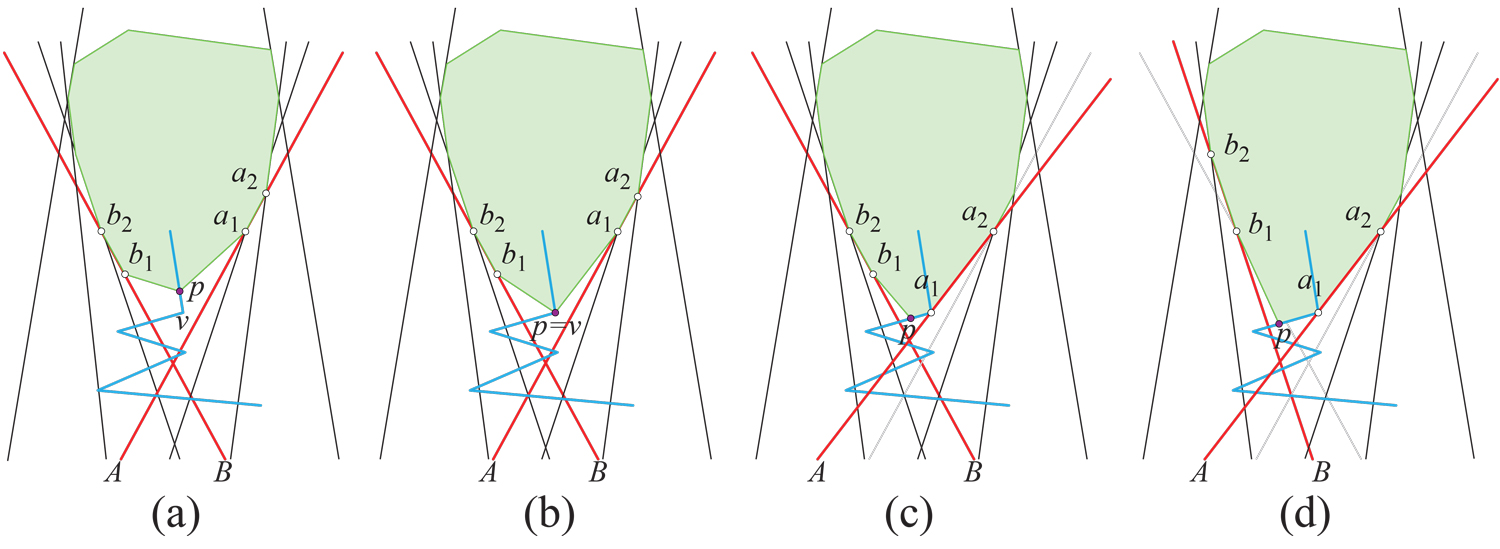}
\caption{Algorithm snapshots.
(a)~$H$ grows without combinatorial change until $p$ reaches $v$.
(b)~$p=v$ event.
(c)~$a_1,a_2$ updated. $e_i$ crosses $B$.
(d)~$b_1,b_2$ updated.}
\figlab{GrowingHull_4steps}
\end{figure}
}

\full{
\begin{algorithm}[htbp]
\caption{Threadable Curve Algorithm: Upper Hull}
\DontPrintSemicolon 
    \SetKwInOut{Input}{Input}
    \SetKwInOut{Output}{Output}

    \Input{Polygonal chain $C=\{v_1,\ldots,v_n\}$}
    \Output{Upper convex hull $H$}
    
    \BlankLine
    \tcp{$p$: Moving point on edge $e_i=(v_{i-1} v_i)$. Fig.~\figref{GrowingHull_4steps}(a).}
    \tcp{$H$: Upper convex hull of $\{v_1,\ldots, v_{i-1}, p\}$.}
    \tcp{$a_1,a_2$: Vertices of $H$ right-adjacent to $p$.}
    \tcp{$b_1,b_2$: Vertices of $H$ left-adjacent to $p$.}
    \tcp{$A$: line through $a_1 a_2$.}
    \tcp{$B$: line through $b_1 b_2$.}
    \BlankLine
    
     \While{$p$ has not reached last vertex $v_n$}{
     
         Compute next event on $e_i$: Intersect $e_i$ with $A$ and $B$.
         
         \If(\tcp*[h]{Fig.~\figref{GrowingHull_4steps}(b)})
         {
         Next event is vertex $v=v_i$. }
         {
         {\If{Turn at $p=v$ enters $\triangle a_1, v, b_1$ and so enters $H$}
         {\Return NotThreadable}
         }
         { \If{Turn at $p=v$ angles outside $H$, so next edge $e_{i+1}$ is on $H$}
         {Update $A$ or $B$.
          \tcp{Fig.~\figref{GrowingHull_4steps}(c).}
          }
         }
         }
         
         \If{Next event is intersection with $A$ or $B$}
         {Update $A$ or $B$, whichever intersected.
         
         \tcp{Fig.~\figref{GrowingHull_4steps}(d).} 
         }  	  
 
     }
     \BlankLine

\end{algorithm}
}

We now walk through the algorithm, whose pseudocode is displayed%
\full{as Algorithm~1.}
\conf{in the full version.}
Let $p$ be on the interior of an edge $e_i=(v_{i-1} v_i)$.
The portion of $e_i$ already
passed by $p$
must lie inside $H$, and the remaining portion outside $H$.
As long as $p$ remains within the wedge region delimited by $A$, $B$,
and $\bH$, the combinatorial structure of $H$ remains fixed
(Fig.~\figref{GrowingHull_4steps}a).
If $p$ crosses $A$ or $B$---say $A$---then $a_1$ leaves $H$ and
$a_1,a_2$ become the next two vertices counterclockwise around $\bH$.
If $p$ reaches the endpoint $v_i$ of $e_i$, then
if $e_{i+1}$ angles outside $H$, $v_i$ becomes a new $a_1$ or $b_1$
depending on the direction of $e_{i+1}$.
If instead, $e_{i+1}$ turns inside $H$, advancing $p$ would enter $H$
and we have detected that $C$ is not threadable by Lemma~\lemref{HullInt}.

All the updates just discussed are constant-time updates:
detecting if $e_i$ crosses $A$ or $B$, updating 
$a_1, a_2$ and $b_1,b_2$,
and detecting if $e_{i+1}$ turns inside $H$, entering
$\triangle b_1 v_i a_1$. 

At the end of the algorithm, $H$ is the hull of $C$.
It may seem surprising that we can compute the hull of $C$
in linear time (rather than $O(n \log n)$), but Melkman showed
long ago that the hull of any simple polygonal chain can be computed
\full{
in linear time~\cite{melkman1987line}.\footnote{
See Dan Sunday's description: 
\url{http://geomalgorithms.com/a12-_hull-3.html}.
See also~\cite{levcopoulos2002adaptive} for results on hulls of 
self-intersecting chains.}
}
\conf{
in linear time~\cite{melkman1987line}.
}
The chain $C$ acts almost as a pre-sorting of the points.


\subsection{Rigid Motions}
\seclab{Motions}
At any stage where the butterfly $\butf(p)$ is non-empty,
we could choose the line $L$ to bisect $\butf(p)$.
This choice was used to produce the online animations cited in Section~\secref{Introduction}.
To prepare for an analogous 3D-computation in Section~\secref{3D},
we explain the bisection choice in terms of vectors normal to $L$.
Fig.~\figref{NormalVecs2D}(a) shows the possible $L$ choices through $p$ dictated by
the two incident edges of $H^+$ and the two incident edges of $H^-$, illustrated by a rightward rays
from $p$ along $L$.
Rotating these $90^\circ$ in~(b) of the figure yields the possible vectors normal to $L$.
The intersection of the $H^+$ and $H^-$ constraints yields an interval
corresponding to $\butf(p)$, which is then bisected to select a particular $N$ and therefore $L$.
(The intersection always yields an interval [rather than two intervals] because each of the $H^+$ and $H^-$ constraints is $\le$ a semicircle.)
\begin{figure}[htbp]
\centering
\includegraphics[width=0.95\linewidth]{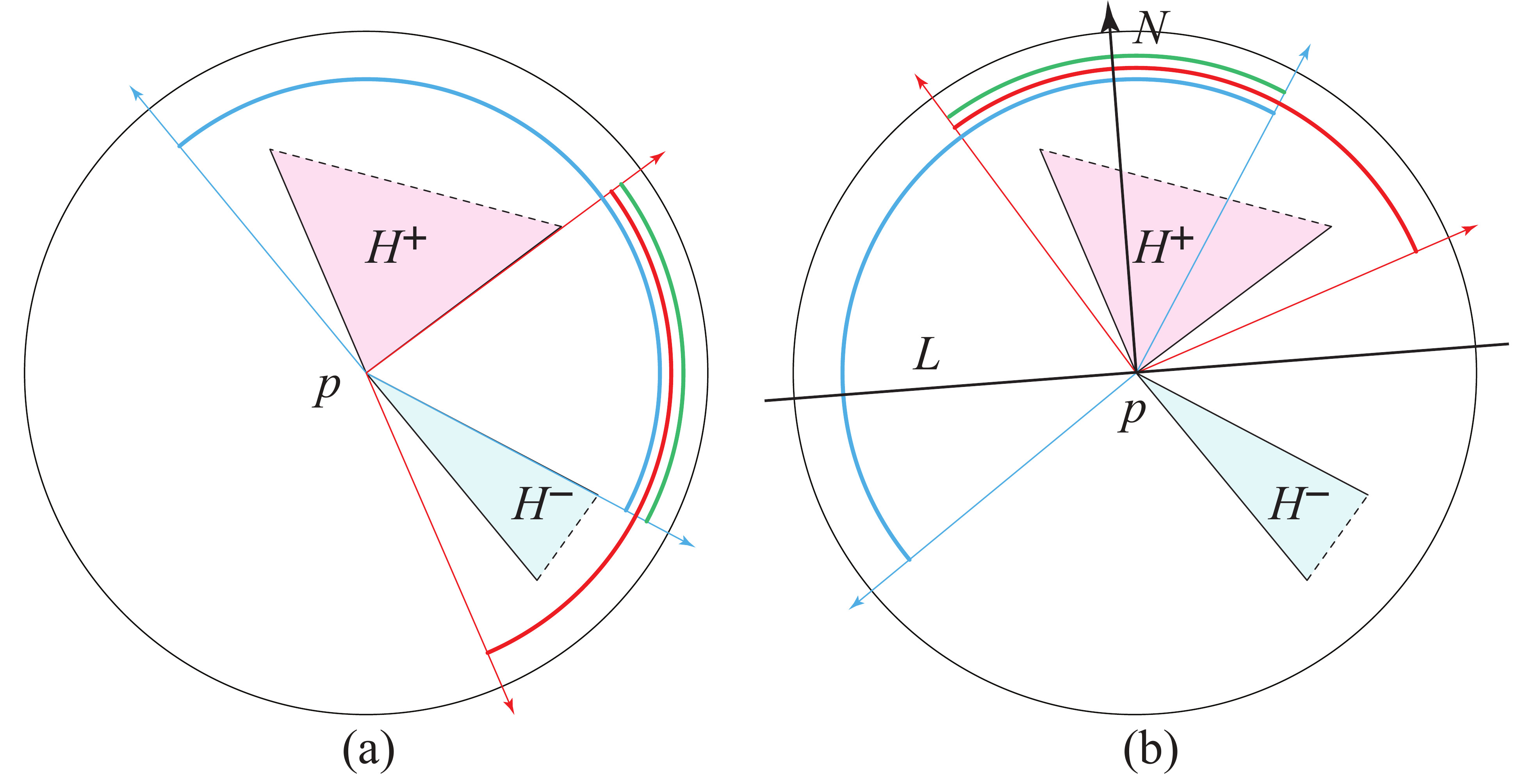}
\caption{(a)~Rightward rays along possible $L$'s. (b)~Normal vector $N$ to $L$ (black).}
\figlab{NormalVecs2D}
\end{figure}

Let $H_j^+$ and $H_j^-$, $j=1,\ldots,m$ be the sequence of hulls at the points
at which there is a combinatorial change in either.
Let $r_j \subseteq e$ be the range of $p$ along edge $e$ of $C$ between
$\{ H_j^+, H_j^- \}$ and $\{ H_{j+1}^+, H_{j+1}^- \}$.
Then as $p$ moves along $r_j$, the wings of the butterfly $\butf(p)$
have the same set of tangency points on the hulls.
With $L$ chosen as the bisector of $\butf(p)$,
the range of $p$ along $r_j$
leads to a translation of $p$ along $r_j$ and a rotation of $L$.
It is not difficult to see that the rotation implied by $p$
moving along $r_j$ reverses at most once,
from clockwise to counterclockwise or vice versa.
This is evident in Fig.~\figref{Theta1Max},
where the butterfly angle $\q$ bisected to yield $L$ has at most one
local maximum.
Thus each slide of $p$ along $r_j$ leads to at most two monotonic rotations.
We call a slide and a simultaneous monotonic rotation
an \emph{elementary rigid motion}.
But note that, although ``elementary,'' these motions are not
pure rotations and pure translations, but rather the particular mix
determined by the slide and the butterfly bisection.
We leave these elementary motions as
the output rigid motions, not further analyzed into explicit 
analytical expressions.
\begin{figure}[htbp]
\centering
\includegraphics[width=0.95\linewidth]{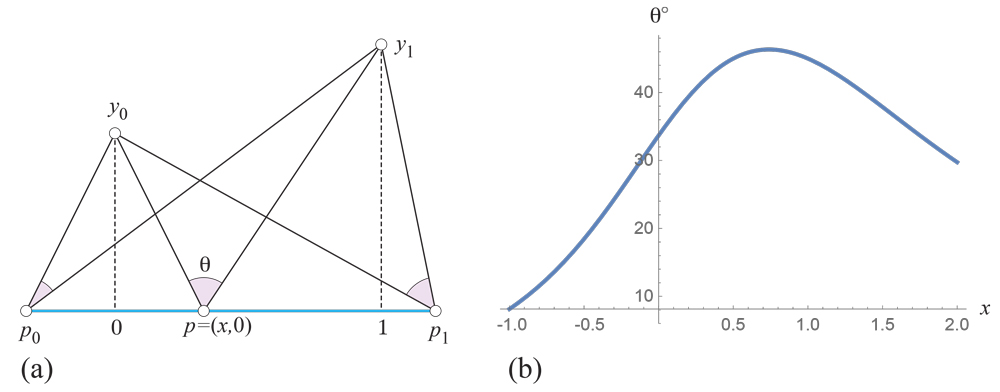}
\caption{(a)~$p$ slides along edge $e$ from $p_0$ to $p_1$,
$r_j=(p_0,p_1) \subseteq e$.
(b)~The butterfly angle $\q$ has at most one local maximum throughout
the range.}
\figlab{Theta1Max}
\end{figure}

Thus the sequence of hulls $O(n)$ provides a set of $O(n)$ elementary rigid motions to thread $C$,
which we used to produce the online animations. 

\subsection{Difficult-to-Thread Curves}
\seclab{Worst}
One easy consequence of our analysis is that a threadable curve need never 
``back-up'' while threading through a hole,
because $p$ never enters $H^{\pm}$ as it progresses along the chain.
However, one could define the ``difficulty'' of threading by, say, 
integrating the absolute value of the back-and-forth rotations necessary to thread.
Then variations on the curve shown in 
Fig.~\figref{ThreadableCurveWorst} are difficult to thread in this sense.
For each pair of adjacent spikes require a rotation by $\q$,
and with many short spikes, there is no bound on $\sum |\q|$
even for a fixed-length chain.\footnote{Thanks to Anna Lubiw for this observation.}
\begin{figure}[htbp]
\centering
\includegraphics[width=0.75\linewidth]{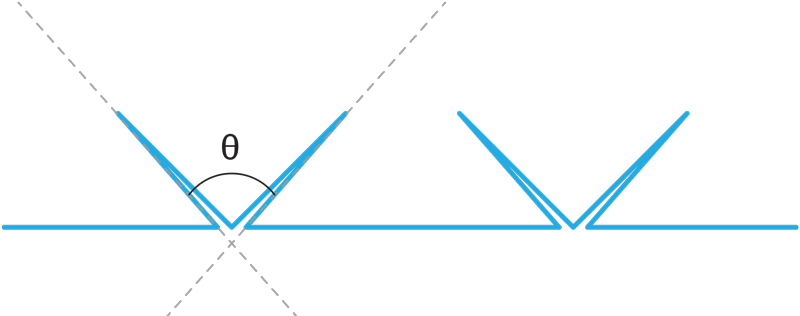}
\caption{A threadable curve that requires repeated rotations.
Animation: \protect\url{http://cs.smith.edu/~jorourke/Threadable/}, Example~2.}
\figlab{ThreadableCurveWorst}
\end{figure}

\section{Algebraic Curves}
\seclab{Algebraic}
\conf{In the full version, we sketch an argument that shows detection of threadability for
algebraic curves is achievable in time polynomial in the degree of the curve.
}
\full{Here we sketch an argument that shows detection of threadability for
algebraic curves is achievable in time polynomial in the degree of the curve.
We avoid computing the hulls, the main tool in the algorithm just presented,
and instead rely on bi-tangents.
We use this lemma:
\begin{lemma}
Let $C$ have a non-empty butterfly at $p_1 \in C$,
and an empty butterfly at  $p_2 \in C$.
Then for some $p^* \in C$ between $p_1$ and $p_2$,
$bf(p^*)$ is empty and the wing-lines coincide in a
line $L$ that is tangent to $C$ at two (or more) points.
\lemlab{bfDisappear}
\end{lemma}
\begin{proof}
The existence of $p^*$ follows from the continuity of the butterflies:
As a point $p$ moves from $p_1$ to $p_2$,
the non-empty butterfly at $p_1$ must disappear before $p_2$ is reached.
Let $p$ be close to the disappearing point $p^*$,
with $\butf(p)$ non-empty with wings $w^+$ and $w^-$.
Each of $w^+$ and $w^-$ must be tangent to $C$ at a point, and the two
tangency points must be distinct.
As $p$ approaches $p^*$, at some stage these tangency points will
no longer discontinuously change. Then at $p^*$, $A=B=L$ passes through those
limit tangency points, e.g., $t_1$ and $t_2$
in Fig.~\figref{ButterFliesSmoothCurve}.
\end{proof}


\begin{figure}[htbp]
\centering
\includegraphics[width=0.6\linewidth]{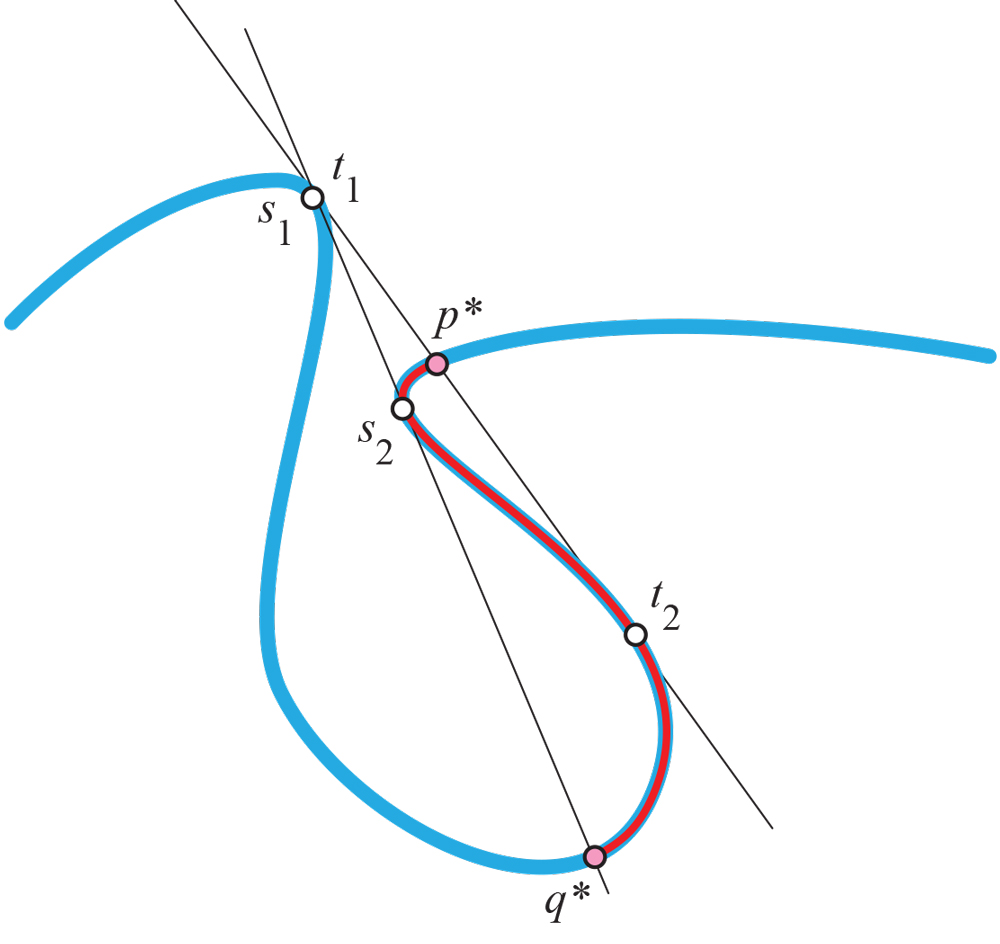}
\caption{A non-threadable smooth curve. Red section has no butterflies.
Both $\butf(p^*) = \varnothing$ and  $\butf(q^*) = \varnothing$.
$t_1,t_2$ and $s_1,s_2$ are the wing-line tangency points,
for $p^*,q^*$ respectively.}
\figlab{ButterFliesSmoothCurve}
\end{figure}

This lemma allows us to detect threadability by checking
all the double tangencies (\emph{bi-tangents}) of $C$, as follows.
Let $L$ be a bi-tangent of $C$, tangent at $t_1$ and $t_2$.
If $L$ does not cross $C$ at some other point $p$, then it
is irrelevant to threadability.
Suppose $L$ does cross $C$ uniquely at $p$.
Then check whether or not this implies an empty $\butf(p)$.
This depends on whether $p$ is between $t_1$ and $t_2$
($p^*$ in Fig.~\figref{ButterFliesSmoothCurve})
or outside those tangencies along $L$ ($q^*$ in the figure), and whether $C$ is locally
left or right at the tangency points.

If, for every bi-tangent $L$, and every corresponding crossing $p$, 
$\butf(p)$ is non-empty, then $C$ is threadable. Otherwise, it is not threadable.
The time-complexity of this algorithm is dependent on the number
of bi-tangents. The other computations (intersecting $L$ with $C$, whether
$C$ is left or right at a tangency) are achievable within the degree of $C$.

It is known that the number of bi-tangents to a curve of
algebraic degree $d$ is $O(d^4)$~\cite{MSE2},
and they can be listed in that time.
So, without delving into details, we can see that the threadability
of an algebraic curve can be decided in time polynomial in the degree of the curve.
}

\section{Threadable Curves in 3D}
\seclab{3D}
The results in Section~\secref{Algorithm}
can be extended to $\R^3$,
asking whether a 3D polygonal chain $C$ can pass through a point-hole in a plane.
First we roughly sketch an algorithm.

Again Lemma~\lemref{HullInt} is the key: we need that $H^+  \cap H^- = \{p\}$
holds for all $p$ on $C$. Again computing $H^+$ and $H^-$ will suffice to 
answer all questions;
see Fig.~\figref{3D_Chain}.
\begin{figure}[htbp]
\centering
\full{\includegraphics[width=0.75\linewidth]{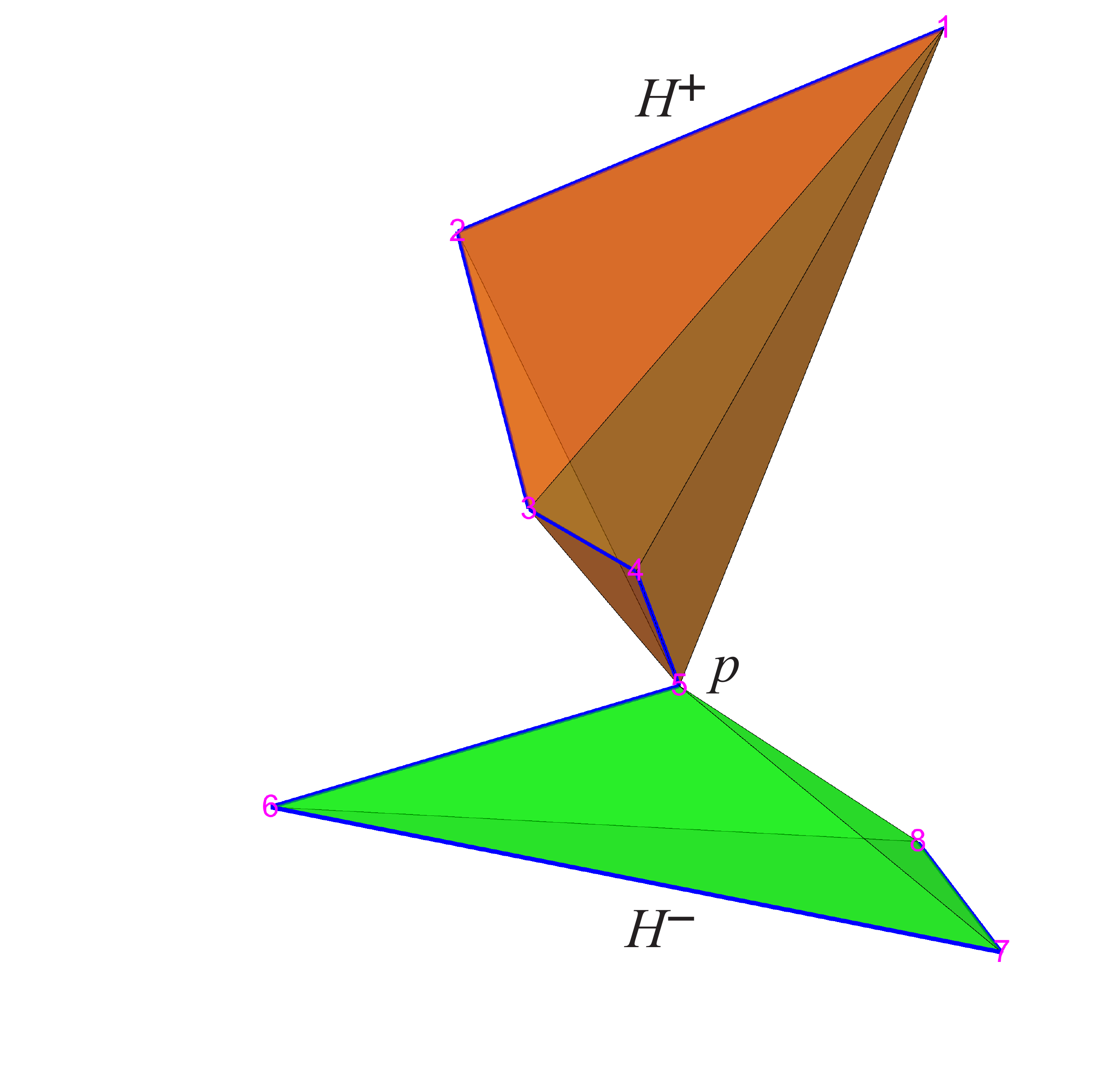}}
\conf{\includegraphics[width=0.6\linewidth]{Figures/3D_Chain}}
\caption{Upper and lower hulls for a 3D polygonal chain.
Animation: \protect\url{http://cs.smith.edu/~jorourke/Threadable/}, Example~6.}
\figlab{3D_Chain}
\end{figure}
But now what was the simple wedge region between hull supporting
lines $A$, $B$,
and $\bH$, becomes a more complex region $R$ bounded by $O(n)$
hull-supporting planes, and the portion of $\bH$ formed by 
the triangles incident to $p$, i.e., what is called
star$(p)$ in simplicial-complex theory (which has size $O(n)$).
Setting aside complexity issues temporarily,
the next edge $e_{i+1}$ on
which $p$ will travel must be intersected
with the planes bounding this region $R$,
to determine whether $R$ changes combinatorially, and if
so, which supporting plane is first pierced by $e_{i+1}$.

The planes bounding $R$ that are not determined by faces
in star$(p)$ are the planes incident to an edge of link$(p)$,
i.e., the edges of star$(p)$ not incident to $p$, which form a topological
circle. See Fig.~\figref{IcosaFig}.
When $e_{i+1}$ pierces a plane $A$ supporting face $\triangle abc$
of $H$, with $ab$ an edge of
link$(p)$, then $ab$ is deleted from the link, and $ac$ and $cb$ added,
and the planes incident to these new link edges are added to those
defining $R$.

\begin{figure}[htbp]
\centering
\includegraphics[width=0.5\linewidth]{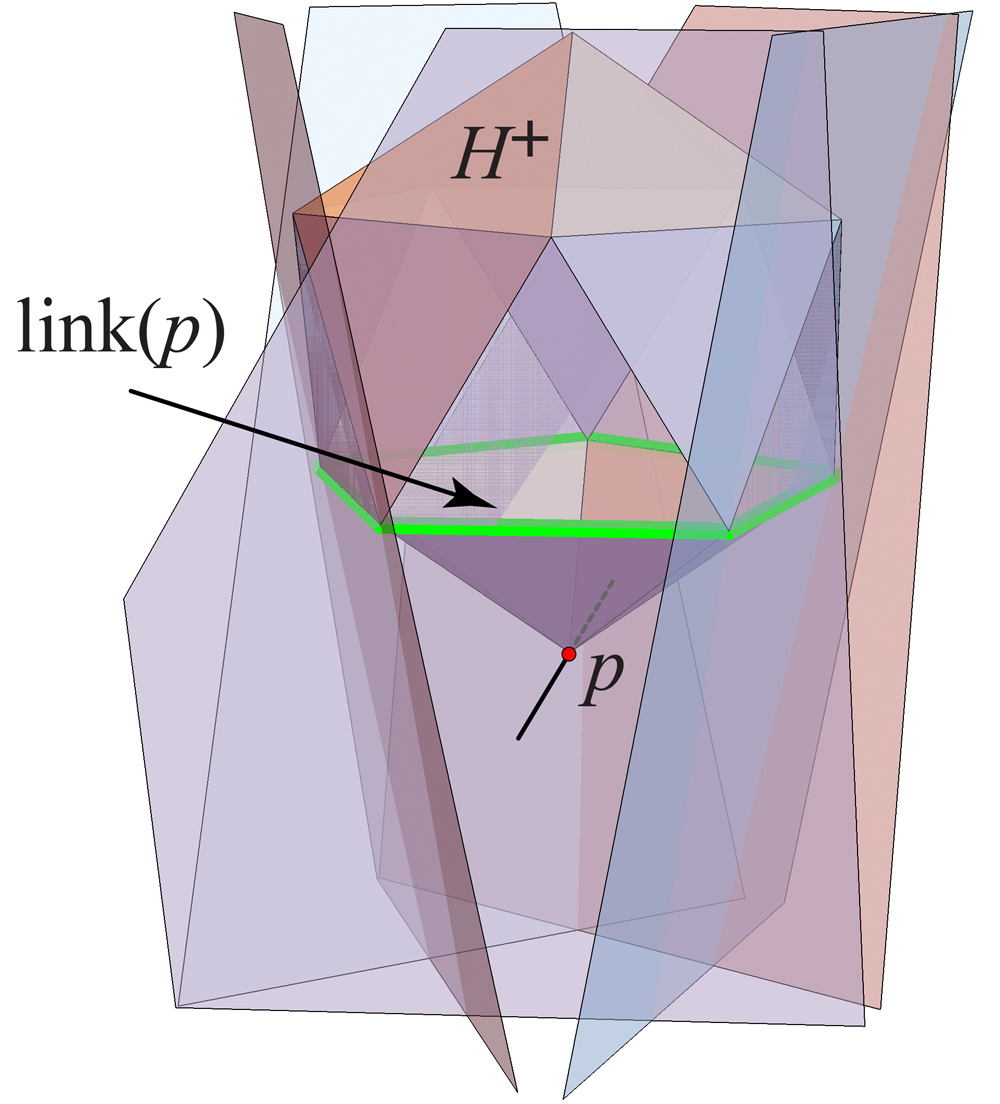}
\caption{Faces sharing an edge with link$(p)$ are extended to form the
lower part of $R$.
$H^+$ does not change combinatorially until $p$ crosses one of
those planes.}
\figlab{IcosaFig}
\end{figure}

This allows $H$ to be maintained throughout the movement
of $p$ along $C$.
As in 2D,  $C$ is threadable if and only if $p$ never enters either hull.

If $C$ is threadable, selecting planes in the more complex
$\butf(p)$ regions and determining rigid motions that achieve the threading are more
complicated tasks than in 2D.

\subsection{Updating the hull $H$ quickly}
\seclab{TChan}
Timothy Chan's powerful dynamic data structure for updating 3D convex hulls~\cite{chan2010dynamic}
provides the tools needed to update the hull $H$ quickly.
Here ``quickly'' means in amortized expected $O(\operatorname{polylog} n)$ time.
His ``nonvertical ray shooting'' queries permit determining if
the next edge $e_{i+1}$ intersects a supporting plane of the region $R$
described above, and if so, which one.
Then that plane can be deleted, and new planes inserted according
to the new link$(p)$,
as identified above.
Thus the computation of the hulls $H^+$ and $H^-$---and therefore
threadability detection---can be
achieved in $O(n \operatorname{polylog} n)$ time.

\subsection{Butterfly ``bisecting'' planes}
\seclab{3DPlanes}
The equivalent of the butterfly $\butf(p)$ in 3D is a more complicated
region than in 2D, and choosing a plane $P$ through $p$ separating
$H^+$ and $H^-$ (the analog of $L$) is correspondingly more complicated.
As in 2D, we identify $P$ by its normal vector $N$, say, pointing toward $H^+$.
The outward normals to the faces of $H-$ incident to $p$ form a convex
geodesic polygon on the Gaussian sphere, with each node a face normal,
and each geodesic arc corresponding to the dihedral angle along the edge shared
by two adjacent faces. See, e.g., \cite{bls-ctelc-07}. Any point within this geodesic
polygon corresponds to a normal vector whose plane supports $H^-$ at $p$.
Repeating this for $H^+$ yields another geodesic polygon corresponding
to the faces of $H^+$ incident to $p$. Using outward face normals leads to
normals pointing toward $H^-$; reflecting this geodesic polygon through the origin
then orients the normals for $H^+$ and $H^-$ consistently.
See Fig.~\figref{GeodesicArcs}.
Then the butterfly region $\butf(p)$ is determined by the intersection $I$ of these
two geodesic polygons.
\begin{figure}[htbp]
\centering
\includegraphics[width=0.5\linewidth]{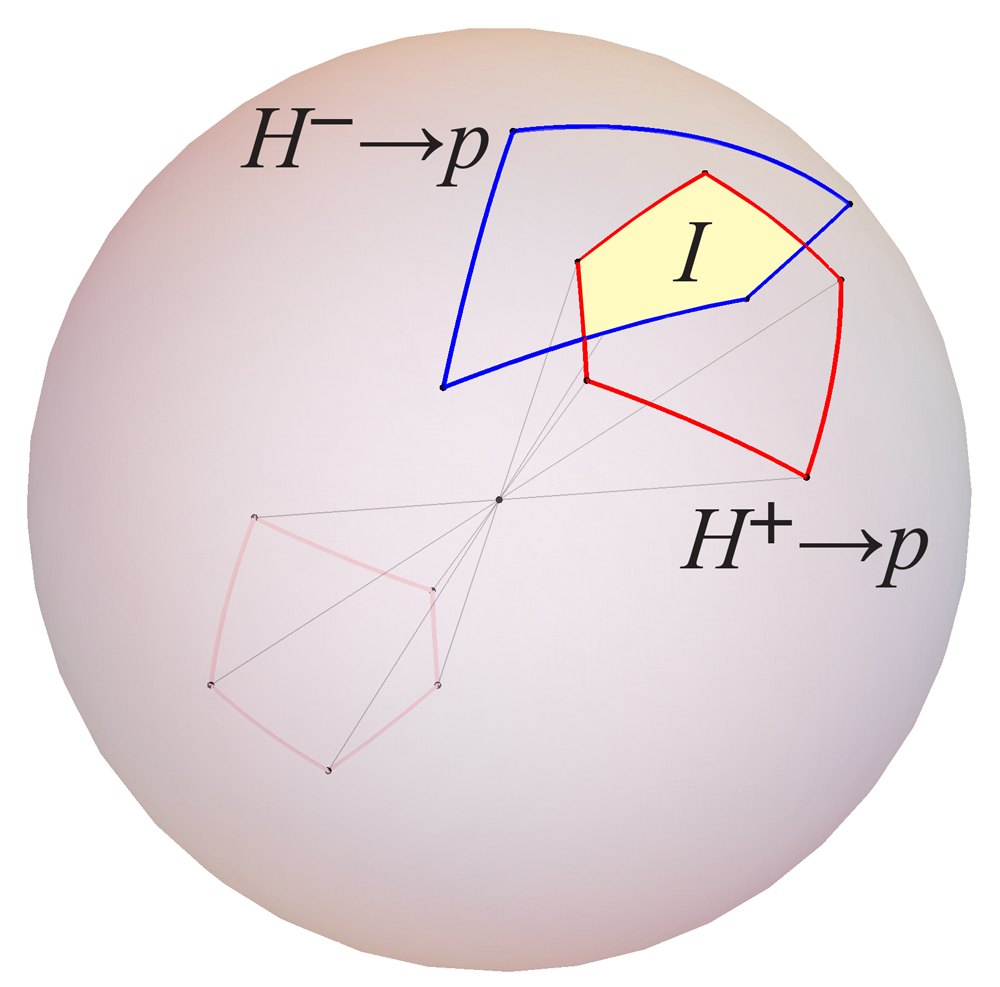}
\caption{Gaussian sphere. The blue polygon represents the faces of $H^-$ incident
to $p$, and the red polygon the faces of $H^+$ incident to $p$. Any point in the
(yellow) intersection $I$ is the normal vector $N$ of a plane $P$ in $\butf(p)$.}
\figlab{GeodesicArcs}
\end{figure}

The equivalent of bisecting $\butf(p)$ in 2D would be choosing the centroid
of the intersection region $I$ on the Gaussian sphere; of course any point in 
the interior of $I$ would suffice.
In 2D we argued that, as $p$ slides along an edge $e$ between combinatorial
changes in either $H^+$ or $H^-$, the rigid rotation reverses direction at most once,
which led to a linear-size description of the rigid motions.
In 3D, even with $p$ on one edge $e$ between combinatorial changes,
it seems that the intersection region $I$ on the Gaussian sphere might
change $\Omega(n)$ times, requiring recalculation of $N \in I$.
This complicates describing the rigid motions in a concise manner.
We leave finding a clean notion of what should constitute
a ``elementary rigid motion'' in 3D to future work,
but we note that the rigid motions
for threading are analytically determined and could be detailed to any precision desired.

\section{Higher Dimensional Generalizations}
\seclab{HigherDim}
There are two natural generalizations to higher dimensions,
but neither seems a fruitful line of future inquiry.
The first retains the curve as a $1$-dimensional object which must pass
through a hole in a hyperplane in $\mathbb{R}^d$. 
\full{The detection algorithm
would again rely on computing the hulls $H^+$ and $H^-$, and is substantially similar
to the analysis in Section~\secref{3D}. We have not pursued this extension.
}

The second generalization replaces the curve with a polygon $P$, which must
pass through a slit in $L$. 
This topic has been explored previously, in two versions.
\conf{We cite~\cite{yap1987move} and~\cite{bose2005generalizing} and
leave further discussion to the full version.
}
\full{A classic motion-planning problem called ``moving a chair through a doorway''
was solved by Chee Yap with an innovative quadratic algorithm~\cite{yap1987move}, which both computed the ``door-width'' of the polygon---the narrowest door through which it could pass---and described the rigid motions necessary to execute the passage.
A second generalization requires the intersection of $P$ and $L$ to always
be a segment; this does not hold for Yap's motions and indeed would
undermine his goals.
Such a polygon is called \emph{sweepable} in~\cite{bose2005generalizing} 
and was explored
as a generalization of monotone polygons.
They provide a quadratic algorithm for detecting sweepability.

Bose et al.~\cite{bose2005generalizing} observed that
sometimes the polygon must ``back up" to pass through the slit,
explicit examples of which were provided in~\cite{jo-nmctd-90}.
As mentioned earlier, this is a phenomenon never needed for threadable curves,
and is perhaps the underlying reason threadability is linear and sweepability seems quadratic.
}

\section{Open Problems}
\seclab{Open}
\begin{enumerate}
\squeezelist
\item In $\mathbb{R}^3$, can finding a plane $P$ separating $H^+$ and $H^-$, 
as sketched in
Section~\secref{3D}, be achieved in $O(n \operatorname{polylog} n)$ time?
In other words, can the intersection $I$ of the two geodesic polygons be
maintained in amortized expected $O(\operatorname{polylog} n)$ time?
\item Is there a natural definition of what constitutes an
``elementary rigid motion'' in $\mathbb{R}^3$, and how many such motions
are needed to thread a polygonal curve of $n$ segments?
\item Can a simple, connected algebraic curve of degree $d$
have $\Omega(d^4)$ bi-tangents?
The $O(d^4)$ bound mentioned in Section~\secref{Algebraic}
is achieved for quartics by disconnected, closed zero-sets~\cite{MObitangents}.
%
\item Define the \emph{minimum clearance} for a threadable curve $C$
as the minimum width region above and below $L$ through which 
points of $C$ pass as it
threads through the hole $o$, with width the dimension parallel to $L$.
If $C$ were a rigid pipe (e.g., a hydraulic tube), it would be necessary to ensure the clearance
regions are empty of other objects to avoid collisions.
Finding the minimum requires more careful selection of $L$ in $\butf(p)$,
rather than just using the bisector as we suggest in Section~\secref{Motions}.
The same question may be asked for $\mathbb{R}^3$.
\item If $C$ is not threadable, what is the shortest slit in $L$ through which $C$ could pass?
Or, in $\mathbb{R}^3$, the smallest radius hole in a plane? 
This is close to Yap's door width~\cite{yap1987move}, and perhaps his algorithm
could be modified to solve the 2D problem.
\end{enumerate}

\paragraph{Acknowledgements.}
We thank Mikkel Abrahamsen
for pointing us to~\cite{bose2005generalizing},
and Anna Lubiw and Joseph Mitchell for helpful discussions.

\bibliographystyle{alpha}
\bibliography{Threadable}

\newcommand{\etalchar}[1]{$^{#1}$}
\begin{thebibliography}{AAI{\etalchar{+}}01}

\bibitem[AAI{\etalchar{+}}01]{aaiklr-gsac-01}
Oswin Aichholzer, Franz Aurenhammer, Christian Icking, Rolf Klein, Elmar
  Langetepe, and G{\"u}nter Rote.
\newblock Generalized self-approaching curves.
\newblock {\em Discrete Appl. Math.}, 109:3--24, 2001.

\bibitem[AFM03]{afm-asmpo-02}
Esther~M. Arkin, S{\'a}ndor~P. Fekete, and Joseph S.~B. Mitchell.
\newblock An algorithmic study of manufacturing paperclips and other folded
  structures.
\newblock {\em Comput. Geom. Theory Appl.}, 25:117--138, 2003.

\bibitem[BLS07]{bls-ctelc-07}
Therese Biedl, Anna Lubiw, and Michael Spriggs.
\newblock Cauchy's theorem and edge lengths of convex polyhedra.
\newblock {\em Algorithms and Data Structures}, pages 398--409, 2007.

\bibitem[BVK05]{bose2005generalizing}
Prosenjit Bose and Marc Van~Kreveld.
\newblock Generalizing monotonicity: On recognizing special classes of polygons
  and polyhedra.
\newblock {\em International Journal of Computational Geometry \&
  Applications}, 15(06):591--608, 2005.

\bibitem[Cha10]{chan2010dynamic}
Timothy~M Chan.
\newblock A dynamic data structure for {3-D} convex hulls and {2-D} nearest
  neighbor queries.
\newblock {\em Journal of the ACM (JACM)}, 57(3):16, 2010.

\bibitem[JO90]{jo-nmctd-90}
Susan Jones and Joseph O'Rourke.
\newblock A note on moving a chair through a doorway.
\newblock {\em Algorithms Review}, 1(3):139--149, 1990.
\newblock S. Jones is now S. Dorward.

\bibitem[LLM02]{levcopoulos2002adaptive}
Christos Levcopoulos, Andrzej Lingas, and Joseph~SB Mitchell.
\newblock Adaptive algorithms for constructing convex hulls and triangulations
  of polygonal chains.
\newblock In {\em Scandinavian Workshop on Algorithm Theory}, pages 80--89.
  Springer, 2002.

\bibitem[Mel87]{melkman1987line}
Avraham~A Melkman.
\newblock On-line construction of the convex hull of a simple polyline.
\newblock {\em Information Processing Letters}, 25(1):11--12, 1987.

\bibitem[Nie17]{MSE1}
Marc Nientker.
\newblock Convex hulls change continuously as one point moves continuously.
\newblock Mathematics Stack Exchange, October 2017.
\newblock \url{https://math.stackexchange.com/q/2529897}.

\bibitem[{\O}kl17]{MSE2}
Jan-Magnus {\O}kland.
\newblock Number of double tangents to an algebraic curve of degree $d$.
\newblock Mathematics Stack Exchange, October 2017.
\newblock \url{https://math.stackexchange.com/q/2494999/237}.

\bibitem[O'R17]{MObitangents}
Joseph O'Rourke.
\newblock Number of bitangents to connected algebraic curve.
\newblock MathOverflow, October 2017.
\newblock \url{https://mathoverflow.net/q/284929}.

\bibitem[Yap87]{yap1987move}
Chee-Keng Yap.
\newblock How to move a chair through a door.
\newblock {\em IEEE Journal on Robotics and Automation}, 3(3):172--181, 1987.

\end{thebibliography}
\end{document}